\documentclass[11pt]{article}

\usepackage[letterpaper,margin=1in]{geometry}
\usepackage[utf8]{inputenc} 
\usepackage[T1]{fontenc}    
\usepackage[colorlinks=true]{hyperref}       
\usepackage{url}            
\usepackage{booktabs}       
\usepackage{amsfonts}       
\usepackage{amsmath,dsfont,amsthm,mathtools,amssymb}
\usepackage{subcaption}

\usepackage{pgfplots}
\pgfplotsset{compat=1.18}

\newcommand{\calI}{\mathcal{I}}
\newcommand{\calL}{\mathcal{L}}
\newcommand{\calN}{\mathcal{N}}
\newcommand{\calNdisc}[1]{\mathcal{N}_{#1 \bbZ}}

\newcommand{\calX}{\mathcal{X}}
\renewcommand{\o}{\mathsf{o}}
\renewcommand{\O}{\mathsf{O}}
\newcommand{\snr}{\mathsf{snr}}

\newcommand{\probP}{\mathsf{P}}

\newcommand{\var}{\mathrm{var}}
\newcommand{\bbE}{\mathbb{E}}

\newcommand{\bbR}{\mathbb{R}}
\newcommand{\bbZ}{\mathbb{Z}}
\newcommand{\intR}{\int_{\bbR}}

\renewcommand{\d}{\,\mathrm{d}}

\newcommand{\dmin}{\mathsf{d}_{\min}}
\newcommand{\dist}{\mathsf{d}}

\newtheorem{theorem}{Theorem}

\newtheorem{corollary}{Corollary}
\newtheorem{lemma}{Lemma}
\newtheorem{proposition}{Proposition}

\DeclareMathOperator*{\argmin}{arg\,min}

\hypersetup{
	colorlinks,
	linkcolor=blue,
	citecolor=red,
	urlcolor=blue
}

\usepackage[
backend=biber,
style=alphabetic,
maxbibnames=15,
maxalphanames=5,
minalphanames=3,
doi=false,
isbn=false,
eprint=false,
backref=true,
]{biblatex}
\title{High Signal-to-Noise Ratio Asymptotics of Entropy-Constrained Gaussian 
	Channel Capacity} 

\author{%
	Adway Girish$^*$, Shlomo Shamai (Shitz)$^\dagger$, Emre Telatar$^*$ \\
	{\small    $^*$School of Computer and Communication Sciences, EPFL}\\
	{\small    $^\dagger$Faculty of Electrical and Computer Engineering, Technion}\\
	{\small   \texttt{adway.girish@epfl.ch},\,\texttt{sshlomo@ee.technion.ac.il},\,\texttt{emre.telatar@epfl.ch}}
}

\begin{document}

	\maketitle
	
	\begin{abstract}
		We study the input-entropy-constrained Gaussian channel capacity problem in 
		the asymptotic
		high signal-to-noise ratio (SNR) regime. 
		We show that the capacity-achieving distribution as SNR goes to 
		infinity is given by a discrete Gaussian distribution supported on a 
		scaled integer lattice.  Further, we show that the gap between the input entropy and the capacity decreases to zero exponentially in SNR, and characterize this exponent.
	\end{abstract}

	\section{Introduction}

	Consider an additive Gaussian noise channel with an input-output relationship given by $Y = X + Z$, where $Z \sim \calN(0,1/\snr)$ is a Gaussian random variable, independent of $X$, and $\snr > 0$ is the signal-to-noise ratio (SNR) of the channel.  The capacity of this channel with an average-power constraint on the input is given by
	\begin{equation*}
		C(\snr) = \sup_{\bbE[X^2] \leq 1} I(X; Y) = \tfrac12 \log ( 1 + \snr),
	\end{equation*}
	with the supremum taken over all distributions of $X$ over $\bbR$, achieved uniquely by taking $X \sim \calN(0,1)$.

In modern cloud-based communication systems, the channel input is not generated directly at the transmitter but is instead produced or selected by a remote agent and conveyed to the transmitter over a rate-limited fronthaul link \cite{fronthaul}. Motivated by this, in earlier work~\cite{girish2025lowsnr}, we introduced an entropy constraint in addition to the power constraint, to model the rate-limited link.
In such a setting, the quantity of interest is the entropy-constrained capacity, defined as 
	\begin{equation}
		C_H(h, \snr) = \sup_{\substack{\bbE[X^2] \leq 1,\\H(X) \leq h}} I(X; Y)
		= \max_{\substack{\bbE[X^2] = 1,\, \bbE[X]=0, \\H(X) = h}} I(X; Y), \label{eqn: C_H}
	\end{equation}
	for $h > 0$. The second equality, i.e., existence of a capacity-achieving distribution satisfying the constraints on entropy, mean and variance with equality, can be justified easily and was shown earlier \cite[Proposition~1]{girish2025lowsnr}.
	
	For any finite $h > 0$, we necessarily have that $X$ is discrete, and hence, $C_H(h, \snr) < C(\snr)$.  As the mutual information is upper bounded by the entropy, we also trivially have $C_H(h, \snr) < h$, with the inequality strict for all finite $\snr$.  Though these upper bounds are strict, the gaps $C(\snr) - C_H(h, \snr)$ and $h - C_H(h, \snr)$ decay to 0 as $\snr$ goes to 0 and $\infty$ respectively.  We showed earlier \cite{girish2025lowsnr} that $C(\snr) - C_H(h, \snr) = \O(\snr^4)$ as $\snr \to 0$ for any $h > 0$, implying that the entropy constraint imposes a negligible penalty on the capacity at low SNR, as $C(\snr) \approx \frac{\snr}{2} \gg \snr^4$.  Equivalently, at low SNR, it is possible to closely approximate the Gaussian capacity by using an arbitrarily low entropy.  The optimal distribution turns out to be the one that matches as many initial 
	moments as possible with $W \sim \calN(0,1)$ while satisfying the entropy constraint.
	
	In this paper, we consider the high SNR asymptotics of $C_H$.  Note that $\lim_{\snr\to\infty} C_H(h, \snr) = h$, and this limit of $h$ is attained by any distribution with entropy equal to $h$.
	We are interested in finding the distribution which ``achieves capacity as $\snr \to \infty$'', defined to be the one for which the gap $h - I(X;Y) = H(X \mid Y)$ decays the fastest in $\snr$.  It is a priori not obvious that any distribution achieves the fastest decay or if it is unique.  We show that this is indeed the case, and that this asymptotic optimal distribution is the one that has the largest minimum distance between its discrete atoms, which is in turn given by a discrete Gaussian distribution \cite{canonne2020discrete,stephens2017gaussian} supported on a scaled integer lattice.
	We also show that the gap $h - C_H(h, \snr)$ goes to zero at least exponentially fast in $\snr$, with an exponent that is decreasing in $h$.  We observe that the exponent behaves differently at ``small'' and ``large'' entropies, and give an explicit approximation for each case.  
	
	\subsection{Prior work}
	Several variations of Shannon's classical Gaussian channel capacity \cite{shannon} have been proposed to study the effect of practical limitations on the capacity \cite{afaycal2001fading,lapidoth2009poisson,rf_2025,merhav2025volumebased}.  A common theme is the appearance of discrete distributions as optimal distributions, even when there is no explicit requirement of discreteness, such as with a peak-amplitude constraint \cite{smith1971information,sharma2010transition,dytso2019capacity,thangaraj2017capacity}.
	Discreteness may also be imposed explicitly, such as the cardinality constraint considered by
	Wu and Verd\'u \cite{wu2010impact}. 
	
	The high-SNR regime that we consider was previously studied by Alvarado et al.~\cite{highSNR_alvarado}, who showed that the conditional entropy decays exponentially in SNR and characterized the exponent for discrete, finite constellations in terms of the minimum distance between consecutive atoms.  We extend this to arbitrary discrete distributions, and use it to conclude that the optimal distribution as $\snr \to \infty$ is the discrete Gaussian distribution~\cite{canonne2020discrete,stephens2017gaussian}.
	Being the discrete analogue of the Gaussian distribution, it arises naturally in computer science applications~\cite{sampling}, lattice-based cryptography \cite{Dwarakanath2014,prest:tel-01245066}, probabilistic shaping of codes~\cite{shaping_discrete}, differential privacy~\cite{canonne2020discrete}, to name a few.

	\subsection{Notation}
	Uppercase letters (e.g.\ $X$, $Y$, \dots) denote random variables. We use $\bbE[X]$ to denote the expectation of $X$.
	We refer to a random variable and its probability distribution interchangeably.
	$\calN(\mu, \sigma^2)$ denotes the Gaussian distribution with mean $\mu$ and variance $\sigma^2$, and $\calNdisc{\beta}(\lambda)$ denotes the discrete Gaussian distribution (see Section~\ref{sec: main} for its definition).
	The mutual information between $X$ and $Y$ is denoted by $I(X;Y)$. The entropy of $X$ is denoted by $H(X)$, and the conditional entropy is given by $H(X \mid Y) = H(X) - I(X;Y)$. 
	$\log$ and $\exp$ are inverses of each other and are taken to base $e$, but we write all numerical values in bits, e.g., $h = \log e = 1$ as $\log_2 e \approx 1.44$ bits, simply to make them more interpretable.
	We write $f(x) = \O(g(x))$ if $\lim_x f(x)/g(x) < \infty$ (the limiting value of $x$ will be clear from context) and $f(x) = \o(g(x))$ if the limit is $0$.

	\section{Main results} \label{sec: main}
	
	We say that a distribution $X^*$ is ``\emph{capacity achieving} or \emph{optimal as $\snr \to \infty$}'' if for any distribution $X$, there exists some $\snr_0 > 0$ such that $I(X^*; X^* + Z) \geq I(X; X + Z)$ for all $\snr > \snr_0$, where $Z \sim \calN(0,1/\snr)$ is independent of $X, X^*$.  This $X^*$ can be used to obtain a lower bound on $C_H(h, \snr)$ that is asymptotically tight as $\snr \to \infty$. 
	For any discrete random variable $X$ with atoms $x_i$, let $\dmin(X) \coloneqq  \inf_{i\neq j} |x_{i} - x_j|$, where $x_i$'s are the atoms of $X$, be the \emph{minimum distance} of $X$.  It is worth noting that this infimum is necessarily attained if $X$ has a finite support, but may not be attained if $X$ has an infinite support, in which case we may also have $\dmin(X) = 0$.
	
	We define the \emph{discrete Gaussian distribution} with parameters $\beta,\lambda$ as follows: we write $X \sim \calNdisc{\beta}(\lambda)$ to mean the discrete distribution supported on the lattice $\beta \bbZ = \{\beta n : n \in \bbZ\}$ with probability mass
	\begin{equation*}
		\probP_{X}(x_i) = \exp(-\lambda i^2) / \sum_{j \in \bbZ} \exp(-\lambda j^2) = \exp(-L(\lambda)-\lambda i^2) 
	\end{equation*}
	at $x_i = \beta i$, where $L(\lambda) = \log\big(\sum_{i \in \bbZ} \exp(-\lambda i^2)\big)$.  Note that the atoms of $X \sim \calNdisc{\beta}(\lambda)$ are equally spaced, with a distance $\beta$ between any consecutive atoms. In the following theorem, we show that the capacity-achieving distribution as $\snr \to \infty$ is given by the discrete Gaussian with $\beta,\lambda$ chosen to satisfy the power and entropy constraints with equality.
	
	\begin{theorem}[High-SNR asymptotics of $C_H$] 
		For finite $h > 0$, the distribution achieving capacity as $\snr \to \infty$ is $\calNdisc{\dist_h}(\lambda_h)$, where $\lambda_h$ is the 
		value of $\lambda$ that solves $h = L(\lambda) - \lambda L'(\lambda)$ and $\dist_h = 1/\sqrt{-L'(\lambda_h)}$.
		Further, the capacity is given by the asymptotic expression 
		\begin{equation}
			C_H(h, \snr) = h - \O\left(\exp\left( - \snr\tfrac{\dist_h^2}{8}  \right)\right). \label{eqn: gap_char}
		\end{equation} 
		\label{thm: char}
	\end{theorem}
	\begin{proof}
		See Section~\ref{sec: proofs}.
	\end{proof}
	
	The quantity $\dist_h$ in the above theorem is the largest minimum distance among all distributions satisfying the power and entropy constraints with equality, i.e., $\dist_h = \sup \{\dmin(X) : \bbE[X^2] = 1, H(X) = h\}$, which (as we show in Section~\ref{sec: dmin}), is achieved by the discrete Gaussian.
	We conjecture that this ``capacity-achieving as $\snr \to \infty$'' property can be strengthened to also say  that the sequence of capacity-achieving distributions at finite $\snr$'s converges weakly (i.e., in distribution) to $\calNdisc{\dist_h}(\lambda_h)$, as Wu and Verd\'u~\cite{wu2010impact} show for a cardinality constraint:  the optimal distribution as $\snr \to \infty$ is shown to be an equally spaced uniform distribution and this is the weak limit of the sequence of true capacity-achieving distributions at finite SNRs. 
	
	The quantity $L(\lambda)$ admits tight approximate characterizations in two regimes:  $\lambda \gg \pi$ and $\lambda \ll \pi$ \cite{canonne2020discrete}.
	In Section~\ref{sec: ent_asymp}, we use these approximate expressions for $\dist_h$ at $h \ll L(\pi) - \pi L'(\pi)$ and $h \gg L(\pi) - \pi L'(\pi)$ respectively, together with Theorem~\ref{thm: char}, to obtain the following explicit approximation for $C_H(h, \snr)$ for small and large entropy (still in the high SNR asymptotic regime).  The threshold $L(\pi) - \pi L'(\pi)$ is approximately 0.48 bits.

	\begin{corollary}[Low- and high-entropy approximation of $C_H$ at asymptotic high SNR]
		As $\snr \to \infty$, we approximate the gap $h - C_H(h, \snr)$ as
		\begin{equation}
			\begin{cases}
				\O\left(\exp\big(-\snr \frac{1}{8h}\log\frac{2}h\big)\right) & h \ll  0.48 \text{ bits},\\
				\O\left(\exp\big(-\snr\frac{\pi e}{4}e^{-2h}\big)\right) & h \gg 0.48 \text{ bits}.
			\end{cases} \label{eqn: small_large_ent}
		\end{equation} 
		Further, for $h \ll 0.48 \text{ bits}$, as $\snr \to \infty$, the capacity-achieving distribution is approximately the symmetric three-point distribution with a large mass at the origin. Meanwhile, as $h \to \infty$, we have $\calNdisc{\dist_h}(\lambda_h)$ converges in distribution to $\calN(0,1)$.
		\label{cor: ent_asymp}
	\end{corollary}
	\begin{proof}
		See Section~\ref{sec: ent_asymp}.
	\end{proof}
	
	The above results are illustrated in Fig.~\ref{fig: vs_snr}, where we compare the conditional entropy $H(X \mid Y)$ of the discrete Gaussian with the asymptotic expression from Theorem~\ref{thm: char}, i.e., $\exp(-\snr\frac{\dist_h^2}{8})$, and its approximation at small and large entropies.

	\begin{figure}[t]
		\centering
		\includegraphics[width=0.7\linewidth]{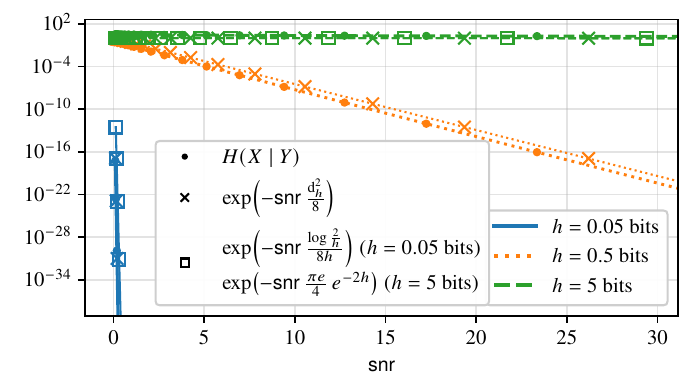}
		\caption{Figure showing $H(X \mid Y)$ for $X \sim \calNdisc{\dist_h}(\lambda_h)$ versus $\snr$, compared against $\exp\big(-\snr \frac{\dist_h^2}{8}\big)$ from Theorem~\ref{thm: char} at three different values of $h$ (in bits): 0.05, 0.5, 5. For the low and high entropy cases, we also compare against the explicit approximations from Corollary~\ref{cor: ent_asymp}. $y$-axis in log-scale.}
		\label{fig: vs_snr}
	\end{figure}
	
	\begin{figure}[t]
		\centering
		\begin{subfigure}[t]{0.5\textwidth}
			\centering
			\includegraphics[width=0.9\linewidth]{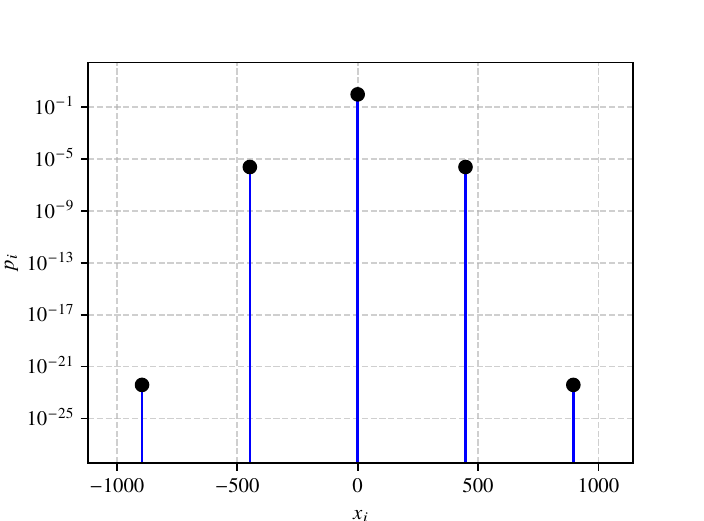}
			\caption{$\calNdisc{\dist_h}(\lambda_h)$ (in log-scale) for $h = 10^{-4}$ bits, $\dist_h = 447.48$}
		\end{subfigure}%
		~ 
		\begin{subfigure}[t]{0.5\textwidth}
			\centering
			\includegraphics[width=0.9\linewidth]{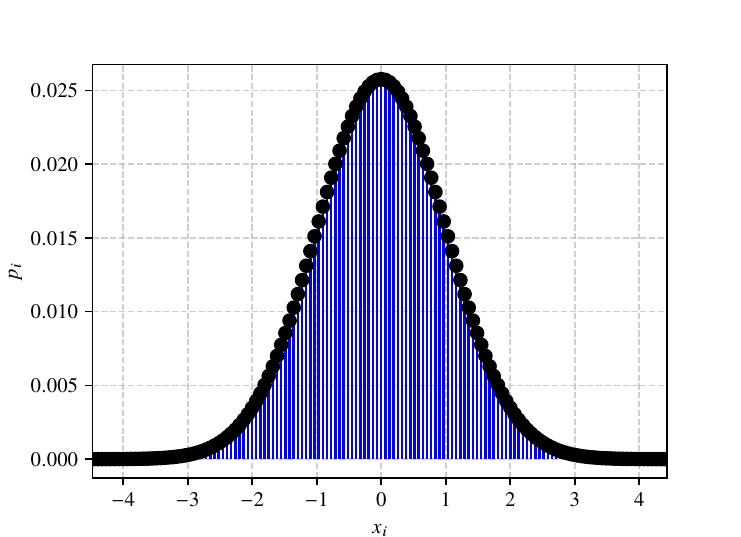}
			\caption{$\calNdisc{\dist_h}(\lambda_h)$ for $h = 6$ bits, $\dist_h = 0.064$}
		\end{subfigure}
		\caption{Comparison of the capacity-achieving distributions as $\snr \to \infty$ for two example distributions with (a) ``small'' entropy, $h = 10^{-4}$ bits: approximately a three-point symmetric distribution with a large mass at the origin (note that the $y$-axis is in log-scale) and (b) ``large'' entropy, $h = 6$: approximately the standard Gaussian distribution.}
		\label{fig: pmfs}
	\end{figure}
	
	\section{Technical details} \label{sec: proofs}
	
	To prove Theorem~\ref{thm: char}, we first show that as $\snr \to \infty$, the conditional entropy $H(X \mid Y) = H(X) - I(X;Y)$ is maximized by the distribution with the largest minimum distance between atoms. (Section~\ref{sec: cond_ent}, Lemma~\ref{lem: cond_ent}). Hence, for a given $h$, the distribution that achieves $C_H(h, \snr)$ as $\snr \to \infty$ is the one with the largest minimum distance among all distributions with variance 1 and entropy $h$. We show that the distribution that achieves this largest minimum distance is the discrete Gaussian (Section~\ref{sec: dmin}). Section~\ref{sec: ent_asymp} then considers the low- and high-entropy regimes and obtains an explicit approximation to $\dist_h$ in each case, which gives us Corollary~\ref{cor: ent_asymp}.

	\subsection{Reduction to minimum distance between atoms} \label{sec: cond_ent}
	
	Let $X$ be an arbitrary discrete random variable having atoms $x_i \in \bbR$
	with probability $p_i > 0$, where $i$ goes over some contiguous subset of $\bbZ$. 
	As $\snr$ goes to $\infty$, we have that $H(X \mid Y)$ goes to zero for any such fixed $X$. The following lemma shows that this decay is exponential in $\snr$ and quantifies the rate as $H(X \mid Y) = \exp\left(-\snr \frac{\dmin(X)^2}{8} + \o(\snr)\right)$.  This was shown earlier~\cite{highSNR_alvarado} for finite support distributions that have $\dmin(X) > 0$, we extend it to arbitrary support discrete distributions with $\dmin(X) \geq 0$.

	\begin{lemma}[Exponent of conditional entropy] For any discrete $X$ with a finite entropy, and $Y = X + Z$ with $Z \sim \calN(0,1/\snr)$ independent of $X$, we have
		\begin{equation*}
			\lim_{\snr \to \infty} - \frac{1}{\snr} \log H(X \mid Y) = \frac{\dmin(X)^2}{8}.
		\end{equation*} \label{lem: cond_ent}
	\end{lemma}
	\begin{proof}[Proof]
		Suppose $X$ is discrete and has entropy $H(X) < \infty$. We will show the above result by deriving upper and lower bounds to $H(X \mid Y)$ that have the same exponential rate in $\snr \eqqcolon \frac{1}{\sigma^2}$. 
		We also use the following notation: let $\varphi_{\sigma}(x)$ denote the density of a Gaussian random variable with mean zero and variance $\sigma^2$, i.e., $\varphi_{\sigma}(x) \coloneqq \frac{1}{\sqrt{2\pi\sigma^2}} \exp(-\frac{x^2}{2\sigma^2})$, and let $\varphi \coloneqq \varphi_1$ denote the standard Gaussian density.
		
		\paragraph{(a) Lower bound to $H(X \mid Y)$. } 
		Fix $\epsilon > 0$. 
		Let $x_0$ and $x_1$ be such that $x_0 < x_1$ and $\Delta = x_1 - x_0 < \dmin + \epsilon$. Note that $x_0,x_1,\Delta$ depend on $\epsilon$ but we do not write this dependence explicitly. Since $H(X \mid Y)$ is the average of $\log\frac1{p_{X|Y}}(x|y)$ (which is always non-negative) over the joint distribution $(X,Y)$, we can lower bound $H(X \mid Y)$ as
		\begin{align*}
			H(X \mid Y) &= \sum_{x \in \calX}  \intR \log \frac{\sum_{x' \in \calX} p_{x'} \varphi_{\sigma}(y-x')}{p_{x} \varphi_{\sigma}(y - x)}\ p_{x}\varphi_{\sigma}(y-x) \d y\\
			&\geq  \intR \log \left(1+\frac{ p_{x_1} \varphi_{\sigma}(y-x_1)}{p_{x_0} \varphi_{\sigma}(y - x_0)}\right) p_{x_0}\varphi_{\sigma}(y-x_0) \d y,
		\end{align*}
		by considering only the $x_0$ term in the outer sum and only the $x_0$ and $x_1$ terms in the sum inside $\log$. 
		We now further lower bound this integral by picking an arbitrary $\delta > 0$ and restricting the domain of integration to the interval $(\frac{x_0 + x_1}{2}, \frac{x_0 + x_1}{2}+\delta)$, i.e.,
		\begin{align*}
			H(X \mid Y) &\geq  \int_{\frac{x_0 + x_1}{2}}^{\frac{x_0 + x_1}{2}+\delta} \log \left(1+\frac{ p_{x_1} \varphi_{\sigma}(y-x_1)}{p_{x_0} \varphi_{\sigma}(y - x_0)}\right) p_{x_0}\varphi_{\sigma}(y-x_0) \d y.
		\end{align*}
		Consider the expression $$\frac{ p_{x_1} \varphi_{\sigma}(y-x_1)}{p_{x_0} \varphi_{\sigma}(y - x_0)} = \frac{ p_{x_1} \exp(-\frac{(y-x_1)^2}{2\sigma^2})}{p_{x_0} \exp(-\frac{(y-x_0)^2}{2\sigma^2})} = \frac{p_{x_1}}{p_{x_0}}\exp\left(\frac{1}{\sigma^2}(x_1-x_0)(y - \tfrac{x_0 + x_1}{2})\right).$$
		This is increasing in $y$ as $x_1 > x_0$, so we can obtain a lower bound to the above expression by substituting $y = \frac{x_0 + x_1}{2}$, which makes the argument of the exponent vanish, giving 
		\begin{align*}
			H(X \mid Y) &\geq p_{x_0}\log \left(1+\tfrac{p_{x_1}}{p_{x_0}} \right) \int_{\frac{x_0 + x_1}{2}}^{\frac{x_0 + x_1}{2}+\delta}  \varphi_{\sigma}(y-x_0) \d y \\
			&= p_{x_0}\log \left(1+\tfrac{p_{x_1}}{p_{x_0}} \right) \int_{\frac{\Delta}{2\sigma}}^{\frac{\Delta}{2\sigma} + \frac{\delta}{\sigma}}  \varphi(y) \d y,
		\end{align*}
		with the second line obtained by making the substitution $\frac{y-x_0}{\sigma} \mapsto y$. 
		The integrand $\varphi(y)$ is decreasing for $y$ in the domain of integration (positive $y$), so we can lower bound the integrand by substituting $y = \frac{\Delta}{2\sigma} + \frac{\delta}{\sigma}$, i.e.
		\begin{align*}
			H(X \mid Y) &\geq p_{x_0}\log \left(1+\tfrac{p_{x_1}}{p_{x_0}} \right) \cdot \tfrac{1}{\sqrt{2\pi \sigma^2}} \exp\left(- \tfrac{(\Delta + 2\delta)^2}{8\sigma^2}\right) \cdot \tfrac{\delta}{\sigma},
		\end{align*}
		for every $\delta, \epsilon > 0$. Hence, the logarithm of the conditional entropy is lower bounded as
		\begin{align*}
			\log H(X \mid Y) \geq \log \left[\tfrac{1}{\sqrt{2\pi \sigma^2}}\tfrac{\delta}{\sigma} p_{x_0}\log \left(1+\tfrac{p_{x_1}}{p_{x_0}} \right) \right] - \tfrac{(\Delta + 2\delta)^2}{8\sigma^2},
		\end{align*}
		which implies that the exponential rate as $\sigma \to 0$ (i.e., $\snr \to \infty$) is 
		\begin{align*}
			\lim_{\sigma \to 0} \sigma^2 \log H(X \mid Y) &\geq \lim_{\sigma \to 0} \sigma^2\log \left[\tfrac{1}{\sqrt{2\pi \sigma^2}}\tfrac{\delta}{\sigma} p_{x_0}\log \left(1+\tfrac{p_{x_1}}{p_{x_0}} \right) \right] - \tfrac{(\Delta + 2\delta)^2}{8}\\
			&= -\tfrac{(\Delta + 2\delta)^2}{8},
		\end{align*}
		for every $\delta, \epsilon > 0$. Letting $\delta \to 0$, we have 
		$$\lim_{\sigma \to 0} \sigma^2 \log H(X \mid Y) \geq -\tfrac{\Delta^2}{8} > -\tfrac{(\dmin + \epsilon)^2}{8},$$
		as $\Delta < \dmin + \epsilon$, by our choice of $x_0, x_1$. Now letting $\epsilon \to 0$, we have the desired result, i.e., $\lim_{\sigma \to 0} \sigma^2 \log H(X \mid Y) \geq -\tfrac{\dmin^2}{8}$. (Note that we do not require $\dmin > 0$ anywhere, the proof also works for $\dmin = 0$.)

		\paragraph{(b) Upper bound to $H(X \mid Y)$. }
		\sloppy
		
		If $\dmin = 0$, then we use the trivial upper bound $H(X \mid Y) \leq H(X)$, as this gives $\lim_{\sigma \to 0}  \sigma^2\log H(X \mid Y) \leq \lim_{\sigma \to 0} \sigma^2 \log H(X) = 0 = -\dmin^2/8$. 
		For the remainder of this upper bound, assume that $\dmin > 0$. 
		Hence, the support is either finite, or if it is infinite, there are no finite accumulation points. 
		
		In either case, let $\calX = (x_i)_{i \in \calI}$ be the set of atoms of $X$, with $\calI = \{1,\dots,n\}$ if $X$ is finite and $\calI = \bbZ$ if $X$ is infinite. We may, without loss of generality, assume that $x_i < x_{i+1}$ for all $i$, as there are no accumulation points. (This may not be possible if the atoms of $X$ have an accumulation point, consider the example $\calX = \{0\} \cup \{-1-1/n : n \in \bbZ\} \cup \{1 + 1/n : n \in \bbZ\}$.)

		By definition, we have $h(Y) = -\intR  f_Y(y) \log f_Y(y) \d y$, with $f_Y(y) = \sum_{x \in \calX} p_x  \frac{1}{\sqrt{2\pi \sigma^2}} \exp\left(-\frac{(y-x)^2}{2\sigma^2}\right) = \sum_{i \in \calI} p_i \frac{1}{\sqrt{2\pi \sigma^2}} \exp\left(-\frac{(y-x_i)^2}{2\sigma^2}\right)$ being the density of $Y$. This can be simplified as
		\begin{align*}
			h(Y) &= -\intR  f_Y(y) \log \left[\sum_{j \in \calI} p_j \tfrac{1}{\sqrt{2\pi \sigma^2}}\exp\left(-\frac{(y-x_j)^2}{2\sigma^2}\right)\right] \d y\\
			&= -\intR f_Y(y) \log \left[\sum_{j \in \calI} p_j \exp\left(-\frac{(y-x_j)^2}{2\sigma^2}\right)\right] \d y - \intR f_Y(y) \log\tfrac{1}{\sqrt{2\pi \sigma^2}}\d y\\
			&= \frac12\log(2\pi\sigma^2)-\intR \sum_{i \in \calI} p_i \tfrac{1}{\sqrt{2\pi \sigma^2}}\exp\left(-\frac{(y-x_i)^2}{2\sigma^2}\right) \log \left[\sum_{j \in \calI} p_j \exp\left(-\frac{(y-x_j)^2}{2\sigma^2}\right)\right] \d y,
		\end{align*}
		as $f_Y$ integrates to 1.  For each $y$, let $i^*(y) \in \calI$ be the index of the atom of $X$ that is closest to $y$, i.e., $i^*(y) = \argmin_i |y - x_i|$ (in case of ties, choose one arbitrarily --- there will only be countably many such $y$'s; a minimizing $i^*$ is also guaranteed to exist for all $y \in \bbR$ as $\dmin > 0$). Also let $d^*(y) = y - x_{i^*(y)}$ and $d_i(y) = y - x_i$. We write simply $i^*$, $d^*$ and $d_i$ when $y$ is clear from context. 
		
		As $\sum_{j \in \calI} p_j = 1$, we have that $\sum_{j \in \calI} p_j \exp\left(-\tfrac{d_j^2}{2\sigma^2}\right)  \leq 1$, and the integrand above is negative.
		Then, we can lower bound $h(Y)$ by considering only the term corresponding to $i = i^*(y)$ in the summation over $i$, to obtain
		\begin{align*}
			h(Y) &= \frac12\log(2\pi\sigma^2) - \intR p_{i^*} \tfrac{1}{\sqrt{2\pi \sigma^2}}\exp\left(-\tfrac{{d^*}^2}{2\sigma^2}\right) \log \left[\sum_{j \in \calI} p_j \exp\left(-\tfrac{d_j^2}{2\sigma^2}\right) \right]\d y.
		\end{align*}
		Note that $H(X \mid Y) = H(X) - I(X;Y) = H(X) - h(Y) + h(Y \mid X) = H(X) + \frac{1}{2}\log(2\pi e \sigma^2) - h(Y)$, as $h(Y \mid  X) = h(Z)$, which gives
		\begin{align}
			H(X \mid Y) &= H(X) + \frac12 + \intR p_{i^*} \tfrac{1}{\sqrt{2\pi \sigma^2}}\exp\left(-\tfrac{{d^*}^2}{2\sigma^2}\right) \log \left[\sum_{j \in \calI} p_j \exp\left(-\tfrac{d_j^2}{2\sigma^2}\right) \right]\d y . \label{eqn: h_x_y}
		\end{align}
		We re-write the integral by splitting $\log \left[\sum_{j \in \calI} p_j \exp\left(-\tfrac{d_j^2}{2\sigma^2}\right)\right]$ as $\log \left[p_{i^*} \exp\left(-\tfrac{{d^*}^2}{2\sigma^2}\right)\right] + \log \left[1 + \sum_{j\in \calI \setminus \{i^*\}} \frac{p_j}{p_{i^*}} \exp\left(\tfrac{{d^*}^2-d_j^2}{2\sigma^2}\right)\right]$. Hence, we have
		\begin{align}
			H(X \mid Y) &= H(X) + \frac12 + \intR p_{i^*} \tfrac{1}{\sqrt{2\pi \sigma^2}}\exp\left(-\tfrac{{d^*}^2}{2\sigma^2}\right) \log \left[p_{i^*} \exp\left(-\tfrac{{d^*}^2}{2\sigma^2}\right)\right]\d y  \nonumber \\
			&\qquad + \intR p_{i^*} \tfrac{1}{\sqrt{2\pi \sigma^2}}\exp\left(-\tfrac{{d^*}^2}{2\sigma^2}\right) \log \left[1 + \sum_{j\in \calI \setminus \{i^*\}} \frac{p_j}{p_{i^*}} \exp\left(\tfrac{{d^*}^2-d_j^2}{2\sigma^2}\right)\right] \d y. \label{eqn: h_x_y_ub}
		\end{align}
		
		To compute these integrals, define $\Delta_i = x_i - x_{i-1}$ for $i \in \bbZ$ when the support is infinite.  When $|\calX| = n$, define $\Delta_i = x_i - x_{i-1}$ for $i = 2,\dots,n$, $\Delta_0 = -\infty$ and $\Delta_n = \infty$.
		Also define $V_i$ to be the open interval $(x_i - \Delta_i/2, x_i + \Delta_{i+1}/2) = (\frac{x_i + x_{i-1}}{2}, \frac{x_{i+1} + x_{i}}{2})$ for $i \in \calI$. Observe that we have $i^*(y) = i$ for all $y$ in $V_i$, and that the $V_i$'s partition $\bbR$.
		We thus compute the first integral in \eqref{eqn: h_x_y_ub}, as
		\begin{align*}
			&\quad \ \intR p_{i^*} \tfrac{1}{\sqrt{2\pi \sigma^2}}\exp\left(-\tfrac{{d^*}^2}{2\sigma^2}\right) \log \left[p_{i^*} \exp\left(-\tfrac{{d^*}^2}{2\sigma^2}\right)\right]\d y\\
			&= \sum_{i \in \calI} \int_{V_i} p_{i} \tfrac{1}{\sqrt{2\pi \sigma^2}}\exp\left(-\tfrac{{d_i}^2}{2\sigma^2}\right) \log \left[p_{i} \exp\left(-\tfrac{{d_i}^2}{2\sigma^2}\right)\right]\d y\\
			&= \sum_{i \in \calI} \left\{ p_i\log p_i \int_{V_i}  \tfrac{1}{\sqrt{2\pi \sigma^2}}\exp\left(-\tfrac{(y-x_i)^2}{2\sigma^2}\right) \d y - p_i\int_{V_i} \frac{(y-x_i)^2}{2\sigma^2}\tfrac{1}{\sqrt{2\pi \sigma^2}}\exp\left(-\tfrac{(y-x_i)^2}{2\sigma^2}\right) \d y \right\}\\
			&= \sum_{i \in \calI} \left\{ p_i\log p_i \left[Q\left(\tfrac{-\Delta_i}{2\sigma}\right)-Q\left(\tfrac{\Delta_{i+1}}{2\sigma}\right)\right] - \frac{p_i}{2}\left[R\left(\tfrac{-\Delta_i}{2\sigma}\right)-R\left(\tfrac{\Delta_{i+1}}{2\sigma}\right)\right] \right\},
		\end{align*}
		where $Q$ is the standard $Q$-function and $R(z) = \int_z^{\infty} \frac1{\sqrt{2\pi}}u^2\exp(-\frac{u^2}{2}) \d u$.
		The $Q$-function satisfies the following useful properties: For $z > 0$, $Q(z)$ satisfies $\frac{z^2}{1+z^2} \frac{1}{\sqrt{2\pi}z}\exp(-\frac{z^2}{2}) \leq Q(z) \leq \frac{1}{\sqrt{2\pi}z}\exp(-\frac{z^2}{2})$. We also have, by definition, that $Q(-z) = 1-Q(z)$.
		The $Q$-function can be related to $R(z)$ by integrating the latter by parts to obtain $R(z) = Q(z) + \frac{z}{\sqrt{2\pi}}\exp(-\frac{z^2}{2})$. Hence, the above expression is equal to 
		\begin{align*}
			&\quad \ \sum_{i \in \calI} \left\{ \left(p_i\log p_i-\tfrac{p_i}{2}\right) \left[Q\left(\tfrac{-\Delta_i}{2\sigma}\right)-Q\left(\tfrac{\Delta_{i+1}}{2\sigma}\right)\right] - \tfrac{p_i}{2\sqrt{2\pi}}\left[\tfrac{-\Delta_i}{2\sigma}\exp\left(\tfrac{-\Delta_i^2}{8\sigma^2}\right)-\tfrac{\Delta_{i+1}}{2\sigma}\exp\left(\tfrac{-\Delta_{i+1}^2}{8\sigma^2}\right)\right] \right\} \\ 
			&= \sum_{i \in \calI} \left\{ \left(p_i\log p_i-\tfrac{p_i}{2}\right) \left[1-Q\left(\tfrac{\Delta_i}{2\sigma}\right)-Q\left(\tfrac{\Delta_{i+1}}{2\sigma}\right)\right] + \tfrac{p_i}{2\sqrt{2\pi}}\left[\tfrac{\Delta_i}{2\sigma}\exp\left(\tfrac{-\Delta_i^2}{8\sigma^2}\right)+\tfrac{\Delta_{i+1}}{2\sigma}\exp\left(\tfrac{-\Delta_{i+1}^2}{8\sigma^2}\right)\right] \right\} \\
			&= -H(X)-\tfrac12 + \sum_{i \in \calI} \left\{ \left(p_i\log \tfrac{1}{p_i}+\tfrac{p_i}{2}\right) \left[Q\left(\tfrac{\Delta_i}{2\sigma}\right)+Q\left(\tfrac{\Delta_{i+1}}{2\sigma}\right)\right] + \right.\\
			&\qquad \qquad \qquad \qquad \qquad \qquad  \qquad \qquad \qquad \qquad \qquad \left.\tfrac{p_i}{2\sqrt{2\pi}}\left[\tfrac{\Delta_i}{2\sigma}\exp\left(\tfrac{-\Delta_i^2}{8\sigma^2}\right)+\tfrac{\Delta_{i+1}}{2\sigma}\exp\left(\tfrac{-\Delta_{i+1}^2}{8\sigma^2}\right)\right] \right\}.
		\end{align*}
		Note that the $H(X) + \frac12$ from \eqref{eqn: h_x_y_ub} cancels. We now compute the second integral from \eqref{eqn: h_x_y_ub}, by using the upper bound $\log(1+z) \leq z$, as
		\begin{align*}
			& \quad \ \intR p_{i^*} \tfrac{1}{\sqrt{2\pi \sigma^2}}\exp\left(-\tfrac{{d^*}^2}{2\sigma^2}\right) \log \left[1 + \sum_{j \in \calI \setminus \{i^*\}} \frac{p_j}{p_{i^*}} \exp\left(\tfrac{{d^*}^2-d_j^2}{2\sigma^2}\right)\right] \d y\\
			&\leq \intR \sum_{j \in \calI \setminus \{i^*\}} p_j \tfrac{1}{\sqrt{2\pi \sigma^2}}\exp\left(-\tfrac{{d_j}^2}{2\sigma^2}\right) \d y\\
			&= \intR \sum_{j \in \calI} p_j \tfrac{1}{\sqrt{2\pi \sigma^2}}\exp\left(-\tfrac{{d_j}^2}{2\sigma^2}\right) \d y - \intR p_{i^*} \tfrac{1}{\sqrt{2\pi \sigma^2}}\exp\left(-\tfrac{{d^*}^2}{2\sigma^2}\right) \d y.
		\end{align*}
		The first term here is equal to 1, as interchanging the integral and summation over $j$ (by Fubini's theorem) gives 
		\begin{align*}
			\intR \sum_{j \in \calI} p_j \tfrac{1}{\sqrt{2\pi \sigma^2}}\exp\left(-\tfrac{{d_j}^2}{2\sigma^2}\right) \d y = \sum_{j \in \calI} p_j \intR \tfrac{1}{\sqrt{2\pi \sigma^2}}\exp\left(-\tfrac{{(y-x_j)}^2}{2\sigma^2}\right) \d y = \sum_{j \in \calI} p_j \bbE_{\tilde Y \sim \calN(x_j, \sigma^2)}[1] = 1.
		\end{align*}
		Meanwhile, the second term can be computed as
		\begin{align*}
			\intR p_{i^*} \tfrac{1}{\sqrt{2\pi \sigma^2}}\exp\left(-\tfrac{{d^*}^2}{2\sigma^2}\right) \d y &= \sum_{i \in \calI} p_i \int_{V_i} \tfrac{1}{\sqrt{2\pi \sigma^2}}\exp\left(-\tfrac{(y-x_i)^2}{2\sigma^2}\right) \d y\\
			&= \sum_{i \in \calI} p_i \left[Q\left(\tfrac{-\Delta_i}{2\sigma}\right)-Q\left(\tfrac{\Delta_{i+1}}{2\sigma}\right)\right]\\
			&= \sum_{i \in \calI} p_i \left[1 - Q\left(\tfrac{\Delta_i}{2\sigma}\right)-Q\left(\tfrac{\Delta_{i+1}}{2\sigma}\right)\right]\\
			&= 1 - \sum_{i \in \calI} p_i \left[Q\left(\tfrac{\Delta_i}{2\sigma}\right)+Q\left(\tfrac{\Delta_{i+1}}{2\sigma}\right)\right],
		\end{align*}
		or equivalently, the second integral from \eqref{eqn: h_x_y_ub} is upper bounded by $\sum_{i \in \calI} p_i \left[Q\left(\tfrac{\Delta_i}{2\sigma}\right)+Q\left(\tfrac{\Delta_{i+1}}{2\sigma}\right)\right]$. Putting all these together into \eqref{eqn: h_x_y_ub}, we have
		\begin{align*}
			H(X \mid Y) \leq \sum_{i \in \calI} \Big\{ \left(p_i\log \tfrac{1}{p_i}+\tfrac{3p_i}{2}\right) \left[Q\left(\tfrac{\Delta_i}{2\sigma}\right)+Q\left(\tfrac{\Delta_{i+1}}{2\sigma}\right)\right] &+ \frac{p_i}{2\sqrt{2\pi}}\Big[\tfrac{\Delta_i}{2\sigma}\exp\left(\tfrac{-\Delta_i^2}{8\sigma^2}\right)\\
			&\quad + \tfrac{\Delta_{i+1}}{2\sigma}\exp\left(\tfrac{-\Delta_{i+1}^2}{8\sigma^2}\right)\Big] \Big\}.
		\end{align*}
		We know that $Q(z)$ is decreasing for all $z$ and $z\exp(-\frac{z^2}{2})$ is decreasing for all $z > 1$. As $\dmin > 0$, for any $\sigma < \dmin(X)/2 \leq \Delta_i/2$ for all $i$, we have that each of the above $Q$ and $\exp$ terms is decreasing in its argument. Hence, we can further upper bound $H(X \mid Y)$ as
		\begin{align*}
			H(X \mid Y) &\leq Q\left(\tfrac{\dmin}{2\sigma}\right)(2H(X)+3) + \frac{1}{\sqrt{2\pi}}\tfrac{\dmin}{2\sigma}\exp\left(\tfrac{-\dmin^2}{8\sigma^2}\right)\\
			&\leq \exp\left(\tfrac{-\dmin^2}{8\sigma^2}\right) \left[H(X) + \tfrac32 + \tfrac1{\sqrt{2\pi}} \tfrac{\dmin}{2\sigma}\right],
		\end{align*}
		using the fact that $Q(z) \leq \frac{1}{2}\exp(-\frac{z^2}{2})$. This gives us the desired upper bound to $\lim_{\sigma \to 0}\sigma^2 \log H(X \mid Y)$, as
		\begin{align*}
			\lim_{\sigma \to 0} \sigma^2 \log H(X \mid Y) \leq \tfrac{-\dmin^2}{8} + \lim_{\sigma \to 0} \sigma^2 \log\left(H(X) + \tfrac32 + \tfrac1{\sqrt{2\pi}} \tfrac{\dmin}{2\sigma}\right) = \frac{-\dmin^2}{8},
		\end{align*}
		which completes the proof.
	\end{proof}

	By Lemma~\ref{lem: cond_ent}, if $\dmin(X_1) > \dmin(X_2)$, then there exists some $\snr_0$ such that $H(X_1 \mid X_1 +Z) > H(X_2 \mid X_2+Z)$ for all $\snr > \snr_0$, where $Z \sim \calN(0,1/\snr)$.  Hence, the capacity-achieving distribution as $\snr \to \infty$ is the one that achieves the largest minimum distance among all distributions with $H(X) = h$ and $\bbE[X^2] = 1$, if such a maximizer exists.  We now consider this optimization problem.
	
	\subsection{Minimum distance with entropy constraint} \label{sec: dmin}

	The following proposition guarantees that there is indeed a unique distribution achieving this largest minimum distance, and characterizes this maximizing distribution.

	\sloppy
	\begin{proposition}[Characterizing distribution with largest $\dmin$] \label{prop: dh}  As earlier, define $L(\lambda) = \log\big(\sum_{j \in \bbZ} \exp(-\lambda j^2)\big)$, and let $\lambda_h$ be the value of $\lambda$ solving $L(\lambda) - \lambda L'(\lambda) = h$. 
		The optimization problem 
		\begin{equation*}
			\dist_h = \sup_{\substack{\bbE[X^2] = 1,\\H(X) = h}} \dmin(X)
		\end{equation*}
		is solved uniquely by $X \sim \calNdisc{\beta_h}(\lambda_h)$ with $\beta_h = 1/\sqrt{-L'(\lambda_h)}$. This immediately implies that
		$\dist_h = \beta_h = 1/\sqrt{-L'(\lambda_h)}$.      
	\end{proposition}
	
	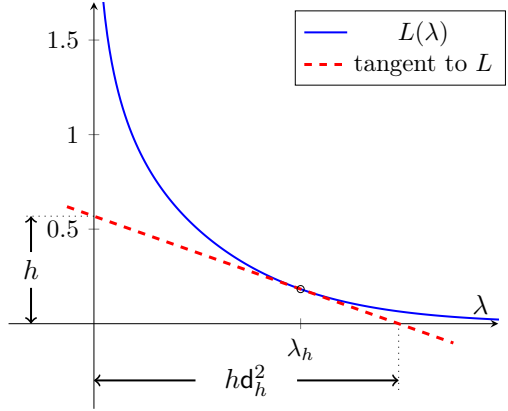
\begin{figure}[thbp]
		\centering
		\begin{tikzpicture}[scale=0.95]
			\begin{axis}[
				axis lines=middle,
				xlabel={$\lambda$},
				ylabel={},
				legend style={font=\small, at={(0.8,0.98)}, anchor=north},
				grid=none,
				ticklabel style={font=\small},
				xmin=-0.95, xmax=4.5,
				ymin=-0.45, ymax=1.7,
				xtick={\xzero},
				ytick={},
				xticklabels={$\lambda_h$}
				]
				
				\addplot+[no markers, solid, thick] table [x=lambda, y=La_a0] {data/La_data.dat};
				\addlegendentry{$L(\lambda)$}

				\def\xzero{2.3}
				\def\yzero{0.18292118389531853}   
				\def\slope{-0.1676711141129772}    
				
				\pgfmathsetmacro{\yint}{\yzero - \slope*\xzero}
				\pgfmathsetmacro{\xint}{\xzero - \yzero/\slope}
				
				\addplot+[no markers, dashed, very thick, domain=-0.3:4, samples=2]
				({x}, {\yzero + \slope*(x-\xzero)});
				\addlegendentry{tangent to $L$}
				\node[circle, draw=black, inner sep=1pt]
				at (axis cs:\xzero,\yzero) {};
				
				\draw[<-, thick] (axis cs:-0.7,0) -- (axis cs:-0.7,\yint/3);
				\draw[draw=none] (axis cs:-0.7,\yint/3) -- (axis cs:-0.7,\yint*2/3)
				node[midway] {$h$};
				\draw[->, thick] (axis cs:-0.7,2*\yint/3) -- (axis cs:-0.7,\yint);
				
				\draw[<-, thick] (axis cs:0,-0.3) -- (axis cs:\xint/3,-0.3);
				\draw[draw=none] (axis cs:\xint/3,-0.3) -- (axis cs:\xint*2/3,-0.3)
				node[midway] {$h \dist_h^2$};
				\draw[->, thick] (axis cs:2*\xint/3,-0.3) -- (axis cs:\xint,-0.3);
				
				\draw[dotted] (axis cs:-0.75,\yint) -- (axis cs:0,\yint);
				
				\draw[dotted] (axis cs:\xint,-0.35) -- (axis cs:\xint,0);
				
			\end{axis}
		\end{tikzpicture}
		\caption{$L(\lambda) = \log\big(\sum_{i \in \bbZ} \exp(-\lambda i^2)\big)$ versus $\lambda$. The $x$- and $y$-intercepts of the tangents to $L$ are $h \dist_h^2$ and $h$ respectively.}
		\label{fig: Llambda}
	\end{figure}
	
	A geometric characterization that may be of interest is the following, shown in Fig.~\ref{fig: Llambda}: $\lambda_h$ is the value of $\lambda$ at which $h = L(\lambda) - \lambda L'(\lambda)$.  The right-hand side is exactly the $y$-intercept of the tangent to $L$ drawn at $\lambda$, hence $\lambda_h$ is the point at which the tangent to $L$ has a $y$-intercept equal to the entropy $h$.  Further, the magnitude of the slope at $\lambda_h$ is $-L'(\lambda_h)$, which is equal to $1/\dist_h^2$, and hence, the $x$-intercept is $h \dist_h^2$.

	\begin{proof} Since there always exists at least one distribution with $\dmin(X) > 0$ satisfying the entropy and second moment constraints, we may assume without loss of generality that the $x_i$'s are ordered.
		The optimization problem in question can be written as
		\begin{equation}
			\begin{aligned}
				\dist_h \quad =  &\sup_{(p_i, x_i)_i} &\inf_{i}\ |x_{i+1} - x_i| \qquad& =&\qquad  \sup_{(p_i, x_i)_i} \quad &\inf_{i}\ |x_{i+1} - x_i|\\
				\qquad &\quad \text{s.t.} \quad  & \sum_{i \in \calI} p_i x_i^2 = 1, \qquad & & \text{s.t.} \quad  & \sum_{i \in \calI} p_i x_i^2 \leq 1,\\
				&& \sum_{i \in \calI} p_i \log \tfrac1{p_i} = h &  && \sum_{i \in \calI} p_i \log \tfrac1{p_i} = h,
				\label{eqn: dh}
			\end{aligned}
		\end{equation}
		with the equality following from the observation that any feasible point $(p_i, x_i)_i$ to the optimization problem on the right can be converted to a feasible point to the problem on the left by simply multiplying each $x_i$ by $\frac{1}{\sqrt{\sum_{i \in \calI} p_i x_i^2}} \geq 1$, which may only increase the objective $\inf_{i}\ |x_{i+1} - x_i|$.
		Further, we can see that the optimization problem \eqref{eqn: dh} is equivalent to 
		\begin{equation}
			\begin{aligned}
				\var_h \coloneqq \qquad   &\inf_{(p_i, x_i)_i} &\sum_{i \in \calI} p_i x_i^2\\
				&\quad \text{s.t.} &\inf_{i} |x_{i+1} - x_i| \geq 1,\\
				&& \sum_{i \in \calI} p_i \log \tfrac1{p_i} = h, 
				\label{eqn: var}
			\end{aligned}
		\end{equation}
		in the sense that the optimal distributions in the two problems are simply scaled versions of each other, as follows. Let $\epsilon > 0$ be arbitrary but smaller than $\dist_h$ (which is necessarily positive for $h > 0$).
		Suppose $(p_i, x_i)_i$ is feasible in \eqref{eqn: dh} and has $\inf_i |x_{i+1} - x_i| > \dist_h - \epsilon > 0$. Replacing each $x_i$ by $\tilde x_i = \frac{x_i}{\dist_h - \epsilon}$, we have a feasible point in \eqref{eqn: var}, which satisfies $\sum_{i \in \calI} p_i \tilde x_i^2 = \frac{1}{\delta^2} \sum_{i \in \calI} p_i x_i^2 \leq \frac{1}{(\dist_h - \epsilon)^2}$.  Then we have  $\var_h = \inf_{(p_i, x_i)_i} \sum_{i \in \calI} p_i x_i^2 \leq \frac{1}{(\dist_h - \epsilon)^2}$, or equivalently, $\var_h (\dist_h - \epsilon)^2 \leq 1$ for arbitrarily small $\epsilon$, and hence $\var_h \dist_h^2 \leq 1$. 
		On the other hand, suppose $(p_i, x_i)_i$ is feasible in \eqref{eqn: var} and has $\sum_{i \in \calI} p_ix_i^2 < \var_h + \epsilon$. Replacing each $x_i$ by $\tilde x_i = \frac{x_i}{\sqrt{\var_h + \epsilon}}$, we have a feasible point in \eqref{eqn: dh}, which satisfies $\inf_i |\tilde x_{i+1} - \tilde x_i| = \frac{1}{\sqrt{\var_h + \epsilon}}\inf_i |x_{i+1} - x_i| \geq \frac{1}{\sqrt{\var_h + \epsilon}}$. Then we have $\dist_h = \sup_{(p_i, x_i)_i} \inf_i |x_{i+1} - x_i| \geq \frac{1}{\sqrt{\var_h + \epsilon}}$, or equivalently, $(\var_h + \epsilon)\dist_h^2 \geq 1$, for arbitrarily small $\epsilon$, which implies that $\var_h \dist_h^2 \geq 1$. Putting the two together, we have $\var_h = \frac{1}{\dist_h^2}$. Further, any solution to \eqref{eqn: dh} can be converted to a solution to \eqref{eqn: var} (or vice-versa) by simply dividing each $x_i$ by $\dist_h$ (or by $\sqrt{\var_h}$, respectively). 
		
		The inequality $\inf_i |x_{i+1} - x_i| \geq 1$ in \eqref{eqn: var} can be strengthened to require $x_{i+1}-x_i = 1$ for all $i$, for the following reason: if $(p_i, x_i)_i$ is such that $\inf_i |x_{i+1} - x_i| \geq 1$, then the distribution with atoms at $\tilde x_i$ given by $\tilde x_{i} = x_{i^*} + (i - i^*)$ with $i^* = \argmin_{i} |x_i|$, necessarily has $|\tilde x_i| \leq |x_i|$ and hence, $\sum_{i \in \calI} p_i \tilde x_i^2 \leq \sum_{i \in \calI} p_i x_i^2$.
		Thus, the solution to \eqref{eqn: var} is evenly spaced unit distance apart, i.e., 
		there exists $a \in [0,1)$ and a contiguous subset $V$ of $\bbZ$ such that 
		the $x_i$ are exactly $a + k$ over all $k \in V$. The solution to \eqref{eqn: dh} is then equally spaced with distance $\dist_h$ such that $\sum_{i \in \calI} p_i x_i^2$ is equal to 1.     
		
		Before proceeding to solve \eqref{eqn: var}, it is useful to note the following relaxation. If we also allow the entropy to be at least $h$ instead of equal to $h$, for each fixed $(x_i)_i$, the optimization problem involves minimizing over $(p_i)_i$, the objective $\sum_{i \in \calI} p_i x_i^2$, which is linear, on a convex set $\sum_{i \in \calI} p_i \log \tfrac1{p_i} \geq h$, whose solution necessarily lies on the boundary. Hence, we solve \eqref{eqn: var}, which is equal to 
		\begin{equation*}    
			\begin{aligned}
				\var_h = \qquad  &\inf_{(p_i, x_i)_i} &\sum_{i \in \calI} p_i x_i^2 & \qquad= &  &\quad \inf_{(p_i)_i, a} &\sum_{i \in \calI} p_i (i + a)^2  \\
				&\ \text{s.t.} & x_{i+1} - x_i = 1 \text{ for all } i,  &  & &\qquad \text{s.t.} & \sum_{i \in \calI} p_i \log \tfrac1{p_i} \geq h,\\
				&& \sum_{i \in \calI} p_i \log \tfrac1{p_i} \geq h &&&& \sum_{i \in \calI} p_i = 1
			\end{aligned}
		\end{equation*}
		Consider first the minimization over $(p_i)_i$.
		This is a convex optimization problem, and hence, the solution $(p_i)_i, a$ must have $\frac{\partial}{\partial p_i} \calL$ equal 0, where $\calL$ is the Lagrangian function with the Lagrange multipliers $\frac1{\lambda} \geq 0$ and $\mu \in \bbR$, given by
		\begin{equation*}
			\calL\big((p_i)_i, a, \lambda, \mu\big) = \sum_{i \in \calI} p_i ( i+a)^2 - \frac1{\lambda} \sum_{i \in \calI} p_i \log \tfrac1{p_i} + \mu \sum_{i \in \calI} p_i.
		\end{equation*}
		The derivative of $\calL$ with respect to $p_i$ is $(i+a)^2 - \frac1{\lambda} \log\frac1{p_i} + \lambda - \mu$, which implies that
		\begin{equation}    
			p_i = \exp(\mu - \tfrac1{\lambda}) \exp(- \lambda(i+a)^2) = \frac{\exp(- \lambda (i+a)^2)}{\sum_{j\in\bbZ} \exp(- \lambda (j+a)^2)} \text{ for } i \in \bbZ, \label{eqn: pi}
		\end{equation}
		by choosing $\mu$ such that $\sum_{i \in \calI} p_i = 1$. Let $L_a(\lambda) = \log \sum_{j \in \bbZ} \exp(- \lambda (j+a)^2)$. Then the entropy of the distribution \eqref{eqn: pi} with $x_i = i+a$ is given by $L_a(\lambda) - \lambda L_a'(\lambda)$ and the variance is given by $- L_a'(\lambda)$, which is to be minimized over all $a \in [0,1)$. That is,
		\begin{equation*}    
			\begin{aligned}
				\var_h =  &\inf_{a \in [0,1)} &  -L_a'(\lambda) \\
				\qquad  &\ \text{s.t.} & L_a(\lambda) - \lambda L_a'(\lambda) = h.
			\end{aligned}
		\end{equation*}
		
		We claim that the solution to the above problem occurs at $a = 0$, via the following geometric argument. The distribution given by \eqref{eqn: pi} with $x_i = i+a$ has entropy $L_a(\lambda) - \lambda L'_a(\lambda)$ and variance $-L'_a(\lambda)$. The former is the $y$-intercept of the tangent to $L_a$ at $\lambda$, and the latter is the magnitude of the slope of the tangent (as $L'_a < 0$). 
		
		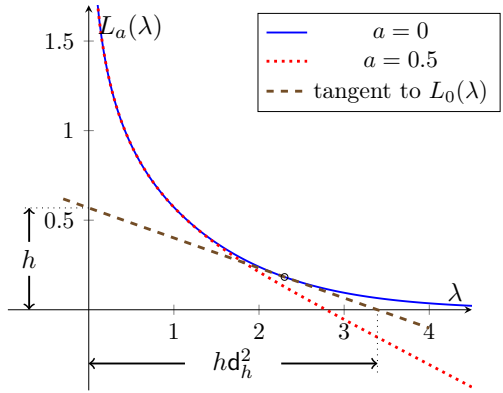
\begin{figure}[!htbp]
			\centering
			\begin{tikzpicture}[scale=0.9]
				\begin{axis}[
					axis lines=middle,
					xlabel={$\lambda$},
					ylabel={$L_a(\lambda)$},
					legend style={font=\small, at={(0.8,0.98)}, anchor=north},
					grid=none,
					ticklabel style={font=\small},
					xmin=-0.95, xmax=4.5,
					ymin=-0.45, ymax=1.7
					]
					\addplot+[no markers, solid, thick] table [x=lambda, y=La_a0] {data/La_data.dat};
					\addlegendentry{$a=0$}
					
					\addplot+[no markers, dotted, very thick] table [x=lambda, y=La_a05] {data/La_data.dat};
					\addlegendentry{$a=0.5$}
					
					\def\xzero{2.3}
					\def\yzero{0.18292118389531853}   
					\def\slope{-0.1676711141129772}    
					
					\pgfmathsetmacro{\yint}{\yzero - \slope*\xzero}
					\pgfmathsetmacro{\xint}{\xzero - \yzero/\slope}
					
					\addplot+[no markers, dashed, very thick, domain=-0.3:4, samples=2]
					({x}, {\yzero + \slope*(x-\xzero)});
					\addlegendentry{tangent to $L_0(\lambda)$}
					\node[circle, draw=black, inner sep=1pt]
					at (axis cs:\xzero,\yzero) {};
					
					\draw[<-, thick] (axis cs:-0.7,0) -- (axis cs:-0.7,\yint/3);
					\draw[draw=none] (axis cs:-0.7,\yint/3) -- (axis cs:-0.7,\yint*2/3)
					node[midway] {$h$};
					\draw[->, thick] (axis cs:-0.7,2*\yint/3) -- (axis cs:-0.7,\yint);
					
					\draw[<-, thick] (axis cs:0,-0.3) -- (axis cs:\xint/3,-0.3);
					\draw[draw=none] (axis cs:\xint/3,-0.3) -- (axis cs:\xint*2/3,-0.3)
					node[midway] {$h \dist_h^2$};
					\draw[->, thick] (axis cs:2*\xint/3,-0.3) -- (axis cs:\xint,-0.3);
					
					\draw[thin, dotted] (axis cs:-0.75,\yint) -- (axis cs:0,\yint);
					
					\draw[thin, dotted] (axis cs:\xint,-0.35) -- (axis cs:\xint,0);
					
				\end{axis}
			\end{tikzpicture}
			\caption{$L_a(\lambda)$ versus $\lambda$ for $ a = 0$ and $a = \frac12$. The $x$- and $y$-intercepts of the tangents to $L_0$ are $h \dist_h^2$ and $h$ respectively.}
			\label{fig: Lalambda}
		\end{figure}
		
		We only need to consider $a = 0$ and $a = \frac12$, as only these give 
		a distribution that has mean zero \cite{canonne2020discrete} --- if the 
		mean is not zero, then translating the distribution to make the mean 
		zero keeps the variance and the entropy the same, while increasing the 
		second moment.
		As can be seen from Figure~\ref{fig: Lalambda}, since $L_0 \geq 
		L_{1/2}$ (proved, e.g., using the Poisson summation formula 
		\cite{canonne2020discrete}), if tangents to $L_0$ and $L_{1/2}$ have 
		the same $y$-intercept, the former must have a larger $x$-intercept. 
		More formally, Lemma~\ref{lem: tangent} lets us conclude that the slope 
		of the tangent to $L_0$ with $y$-intercept $h$ is smaller in magnitude 
		than that to $L_{1/2}$ (Lemma~\ref{lem: tangent} requires that $L_0$ 
		and $L_{1/2}$ are convex and decreasing; that they are decreasing is 
		immediate, while convexity follows from H\"older's inequality). Hence, 
		the variance (given by $-L_a'(\lambda)$ at $\lambda$ such that the 
		$y$-intercept of the tangent is $h$) at $a = 0$ is smaller than the 
		variance at $a = 1/2$, and the minimum variance is obtained by choosing 
		$a = 0$. 
		
		Equivalently, the largest minimum distance in \eqref{eqn: dh} is 
		achieved by choosing $a = 0$ and taking the atoms to be supported on 
		$\frac{1}{\sqrt{-L_0'(\lambda_h)}}$, where $\lambda_h$ is the value of 
		$\lambda$ that solves $L_0(\lambda) - \lambda L'_0(\lambda) = h$. 
		Consequently, the solution to \eqref{eqn: dh}, i.e., the distribution 
		with the largest minimum distance among those with entropy equal to $h$ 
		and second moment equal to $1$, has $p_i$ given by \eqref{eqn: pi} with 
		$a = 0$ and $x_i = \frac{i}{\sqrt{-L_0'(\lambda_h)}}$. This completes 
		the proof, by setting $L \coloneqq L_0$.
	\end{proof}

	\begin{lemma}[Tangents with same $y$-intercept to convex functions] 
		\label{lem: tangent}
		Let $f,g : (0,\infty) \to \bbR$ be two convex, decreasing functions 
		such that $f(x) > g(x)$ for all $x > 0$. Let $\ell_f$ and $\ell_g$ be 
		the tangents to $f$ and $g$ respectively that intersect the $y$-axis at 
		the same point. Then, the slope of $\ell_f$ is strictly smaller in 
		magnitude (less negative) than that of $\ell_g$. 
	\end{lemma}
	\begin{proof}
		Suppose that the slope of $\ell_f$ is strictly larger in magnitude, 
		i.e., more negative, than that of $\ell_g$. Since $\ell_f$ and $\ell_g$ 
		intersect each other on the $y$-axis, we have $\ell_f(x) < \ell_g (x)$ 
		for all $x > 0$. There exists a point $x_f > 0$ where the tangent 
		$\ell_f$ touches $f$, i.e., $\ell_f(x_f) = f(x_f)$.
		Further, since $g$ is convex, it always lies above its tangent 
		$\ell_g$. Hence, we have the chain of inequalities 
		$$g(x_f) \geq \ell_g(x_f) > \ell_f(x_f) = f(x_f),$$
		which contradicts the hypothesis that $f(x) > g(x)$ for all $x > 0$. 
		Thus, the slope of $\ell_f$ must be strictly smaller in magnitude than 
		that of $\ell_g$. 
	\end{proof}

	\subsection{Low- and high-entropy regimes} \label{sec: ent_asymp}
	
	From Section~\ref{sec: dmin}, we have a parametric expression for $\dist_h$ given by $1/\sqrt{-L'(\lambda)}$ at $h = L(\lambda) - \lambda L'(\lambda)$. Fig.~\ref{fig: Llambda} shows the tradeoff between $\dist_h$ and $h$ as $\lambda$ varies from $0$ to $\infty$; in particular, we have $h \to 0$ and $\dist_h \to \infty$ as $\lambda \to \infty$, and $h \to \infty$ and $\dist_h \to 0$ as $\lambda \to 0$.
	
	Consider the function $L(\lambda)$. It is known that $L(\lambda) \approx \frac12\log\frac{\pi}{\lambda}$ for $\lambda \ll \pi$ and $L(\lambda) \approx \log(1 + 2 e^{-\lambda})$ for $\lambda \gg \pi$, and that these approximations are surprisingly tight as long as $\lambda$ is not too close to $\pi$ \cite{canonne2020discrete}. 
	Hence, $\dist_h$ behaves differently depending on whether $h$ is smaller than or larger than $L(\pi) - \pi L'(\pi) \approx 0.48$ bits, as characterized in the lemma below and shown in Fig.~\ref{fig: dh}.

	\begin{figure}[t]
		\centering
		\begin{subfigure}[t]{0.5\textwidth}
			\centering
			\includegraphics[width=\linewidth]{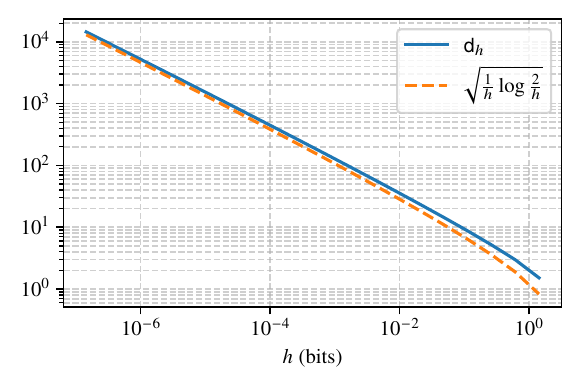}
			\caption{$\dist_h$ versus $h$ for small $h$. Both axes in log-scale.}
		\end{subfigure}%
		~ 
		\begin{subfigure}[t]{0.5\textwidth}
			\centering
			\includegraphics[width=\linewidth]{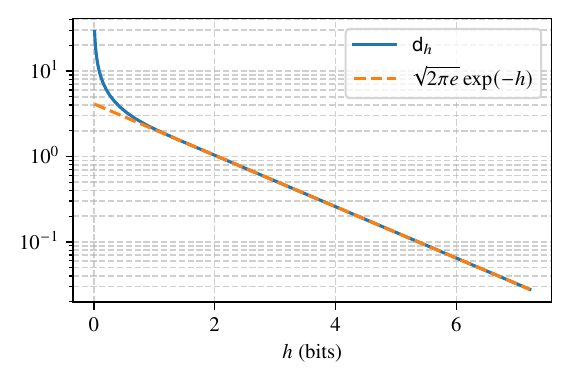}
			\caption{$\dist_h$ versus $h$ for large $h$. $y$-axis in log-scale.}
		\end{subfigure}
		\caption{Comparison of $\dist_h$ and its approximations in the small and large entropy regimes: (a) $\dist_h \approx \sqrt{\frac1h\log\frac{2}h}$ for small $h \ll 0.48$ bits, and (b) $\dist_h \approx \sqrt{2\pi e} \exp(-h)$ for large $h \gg 0.48$ bits.}
		\label{fig: dh} \vspace*{-8pt}
	\end{figure}
	
	\begin{lemma}[Approximation for $\dist_h$ at low and high entropy]
		We have the following approximations for $\dist_h$ at small and large $h$:
		\begin{equation*}
			\dist_h \approx \begin{cases}
				\sqrt{\frac1h\log\frac{2}h} & h \ll L(\pi)- \pi L'(\pi) \approx 0.48 \text{ bits},\\
				\sqrt{2\pi e}\,\exp(-h) & h \gg L(\pi) - \pi L'(\pi). 
			\end{cases}
		\end{equation*} \label{lem: dist_h}
	\end{lemma}
	
	\begin{proof}[Proof]
		We have the following pairs of upper and lower bounds on $L(\lambda)$ that are true for all $\lambda > 0$ but are tight for small and large $\lambda$ respectively \cite{stephens2017gaussian,canonne2020discrete}.  
		For small $\lambda \ll \pi$, we use
		\begin{align*}
			\sqrt{\tfrac{\pi}{\lambda}}\left(1 + 2e^{-\frac{\pi^2}{\lambda}}\right)   \leq  e^{L(\lambda)} \leq \sqrt{\tfrac{\pi}{\lambda}}\left(1 + 2e^{-\frac{\pi^2}{\lambda}}\right) + e^{-\frac{\pi^2}{\lambda}},
		\end{align*}
		giving us the approximation $L(\lambda) \approx \log\sqrt{\frac{\pi}{\lambda}}$,
		and for large $\lambda \gg \pi$, we use
		\begin{align*}
			1 + 2e^{-\lambda} \enspace \leq \enspace e^{L(\lambda)} \enspace \leq \enspace  1 + 2e^{-\lambda} + \sqrt{\tfrac{\pi}{\lambda}}e^{-\lambda},
		\end{align*}
		which gives us $L(\lambda) \approx \log\big(1 + 2e^{-\lambda}\big) \approx 2 e^{-\lambda}$.
		We now prove the claim by using these to approximate $\lambda_h$ and write $\dist_h = 1/\sqrt{-L'(\lambda_h)}$ in terms of $h = L(\lambda_h) - \lambda_h L'(\lambda_h)$ in each regime. The arguments can be made rigorous by replacing the approximations with the appropriate lower and upper bounds at each stage.
		
		For small $\lambda$, we have $L(\lambda) \approx \log\sqrt{\frac{\pi}{\lambda}}$. Hence, $L'(\lambda) \approx -1/(2\lambda)$, which implies that $\dist_h \approx \sqrt{2\lambda_h}$. The entropy is $h = L(\lambda)-\lambda L'(\lambda) \approx \log\sqrt{\frac{\pi}{\lambda}} + \frac12 = \frac12\log\frac{\pi e}{\lambda}$, or equivalently, $\lambda_h \approx \pi e \exp(-2 h)$. Putting the two together, we have $\dist_h \approx \sqrt{2\pi e} \exp(-h)$.
		
		On the other hand, for large $\lambda$, we have $L(\lambda) \approx \log\big(1 + 2e^{-\lambda}\big) \approx 2 e^{-\lambda}$.  Hence, $L'(\lambda) \approx -2 e^{-\lambda}$. The entropy is $h = L(\lambda) - \lambda L'(\lambda) \approx 2 e^{-\lambda} ( 1+\lambda) \approx 2 \lambda e^{-\lambda}$. This is related to the secondary branch of the Lambert-$W$ function $W_{-1}$ \cite{corless1996lambert}, where $W_{-1}(y)$, for small negative $y$, is the large negative $x$ that solves $y = x e^x$. $W_{-1}(x)$ can be approximated as $\log(-x)-\log(-\log(-x))$. Here, we have $-\lambda_h \approx W_{-1}(-\frac{h}{2})$, and hence, $-\lambda_h \approx \log\frac{h}{2} - \log \log \frac{2}{h}$, or equivalently, $\exp(\lambda_h) \approx \frac{2}h\log\frac{2}h$. Hence, $\dist_h = \frac{1}{\sqrt{-L'(\lambda_h)}} \approx \sqrt{\frac{\exp(\lambda_h)}2} = \sqrt{\frac1h\log\frac{2}h}$. 
	\end{proof}
	
	With the above characterization of $\dist_h$, we can now complete the proof of Corollary~\ref{cor: ent_asymp}.
	\begin{proof}[Proof of Corollary~\ref{cor: ent_asymp}]
		We get \eqref{eqn: gap_char} by directly plugging into \eqref{eqn: C_H} the approximations from Lemma~\ref{lem: dist_h}.  
		
		To obtain the expression for $h \ll 0.48$ bits, we used the approximation $L(\lambda) \approx \log(1 + 2e^{-\lambda})$ for large $\lambda$, which is exactly what we obtain if we restrict the summation in $\log\big(\sum_{i \in \bbZ} \exp(- \lambda i^2)\big)$ to include only $i = 0, \pm1$.  This gives the symmetric three-point distribution with a large mass at the origin.
		
		Finally, to check that $\calNdisc{\dist_h}(\lambda_h)$ converges in distribution to $\calN(0,1)$ as $h \to \infty$, note that $\dist_h \approx \sqrt{2 \lambda_h}$.  Hence, it is enough to check that $\calNdisc{a}(a^2/2)$ converges to $\calN(0,1)$ as $a$ goes to 0, which follows from \cite[Proposition~27]{canonne2020discrete}.
	\end{proof}

	\section{Discussion and conclusion}  \label{sec: conc}
	We studied the problem of computing the capacity of a power-constrained Gaussian channel with an additional entropy constraint in the asymptotic high-SNR  regime.  We observed that the asymptotic optimal distribution is the one that has the largest minimum distance between consecutive atoms, showed that the discrete Gaussian uniquely achieves this largest minimum distance, and characterized the exponential rate at which the entropy-constrained capacity approaches $h$ as $\snr \to \infty$.
	
	One claim that remains unresolved is the convergence of the sequence $(X^*_\snr)_\snr$ of true capacity-achieving distributions at finite SNRs to the discrete Gaussian as $\snr \to \infty$.  The main difficulty is that Lemma~\ref{lem: cond_ent} no longer holds in general if $X$ is allowed to change with $\snr$, and indeed, it is possible to find counterexamples by allowing the probability masses at the atoms achieving $\dmin$ to decay exponentially fast in $\snr$.  Nonetheless, we believe that $X^*_\snr$ should still satisfy Lemma~\ref{lem: cond_ent}, i.e., $\lim_{\snr \to \infty} -\frac1\snr \log H(X^*_{\snr} \mid X^*_{\snr} + Z) = \lim_{\snr \to \infty} \dmin(X^*_\snr)^2/8$, where $Z \sim \calN(0,1/\snr)$.  
	This is exactly the setup for which the tools of $\Gamma$-convergence \cite{braides2002gamma} have been developed, but it remains unexplored by the information theory community as yet.  
	
	Since such a result clearly relies on $X^*_\snr$ being the optimal distribution at finite SNR,  the necessary KKT conditions satisfied by the optimal $X$ with atoms $x_i$ of probability $p_i$ might be of use.  They are given by the following self-consistent equations:
	\begin{equation*}
		p_i \propto \exp\Big(\frac{U(x_i) + \lambda {x_i}^2}{\mu}\Big) \quad \text{and} \quad U'(x_i) = -2\lambda x_i
	\end{equation*}
	for all $i$, where $U(x) = \bbE[\log f_Y(x + Z)]$ and $f_Y$ is the density of $Y = X + Z$, and $\lambda,\mu > 0$ are parameters to be chosen to make $\bbE[X^2] = 1$ and $H(X) = h$.
	
	In the absence of an entropy constraint, we would only have had the second condition, and it can be checked that it is satisfied at all $x \in \bbR$ for $X \sim \calN(0,1)$.  With the entropy constraint, however, the parameter $\mu$ acts as a ``temperature'' controlling the spread of the distribution.  Another interpretation is the following: the quantity $U(x) + \lambda x^2$ represents an ``energy'', the atoms are the stationary points of this energy, and the probabilities are given by a softmax of the energies (but $U$ also depends on the probability distribution).  It would be interesting to see if such a line of thinking can yield bounds on the cardinality of the optimal distribution at finite SNR.

	\printbibliography
	
\end{document}